\documentclass[english]{article}
\usepackage[english]{babel}
\usepackage{amsmath,amsthm,amssymb}
\usepackage{latexsym}
\usepackage{enumerate}
\usepackage{hyperref}

\theoremstyle{remark}
\theoremstyle{hypothesis}

\newtheorem{theorem}{Theorem}
\newtheorem{lemma}{Lemma}

\newtheorem*{hypothesis}{(H)}
\newtheorem{remark}{Remark}
\newtheorem{example}{Example}
\newtheorem{definition}{Definition}

\title{Noncommutative integration, quantum mechanics, Tannaka's theorem for compact groupoids and examples}

\author{Artur O. Lopes\thanks{Instituto de  Matematica e Estatistica - UFRGS - Brazil}\;, Marcos Sebastiani\footnotemark[1]\; and Victor Vargas\thanks{Center for Mathematics - University of Porto - Portugal}}

\begin{document}

\maketitle

{\bf Abstract:} We  consider topological groupoids  in finite and also in a compact settings. In the initial sections, we introduce definitions of typical observables and we studied them in the context of statistical mechanics and quantum mechanics. We exhibit explicit examples and one of them will be   the so-called quantum ratchet.  This is related to Schwinger's algebra of selective measurements. Here we consider  $\mathcal{G}$-kernels, transverse functions, modular functions, and quasi-invariant measures for Haar systems.  Later we present our main result  which is a version of Tannaka's theorem for Hausdorff compact groupoids - extending the original proof of T. Tannaka.



\section{Introduction}
\label{introduction-section}

Considering appropriate  definitions in the framework of topological groupoids, our main goal here is  to present some interesting results concerning  finite and Hausdorff compact groupoids. 
We are interested in the study of properties that are potentially applicable to quantum mechanics and dynamical systems.  Schwinger's algebra of selective measurements has a natural interpretation in the formalism of groupoids. In this direction we present some characteristics of observables defined on the so-called ``quantum ratchet" (see  \cite{Ibort2}, \cite{Ibort}, \cite{Ibort1} and \cite{Ibort3} for more details).   Later, we extend that characterizations to the non-finite compact context. We will also describe the so-called $\mathcal{G}$-kernels, transverse functions, modular functions, and quasi-invariant measures for Haar systems in our setting; these are  fundamental concepts in  non-commutative integration.  As far as we know there are not many explicit examples of these entities in the mathematical literature, and we provide them.
One of our main results is the proof of a version of the so-called Tannaka's Theorem in the setting of Hausdorff compact groupoids, which is an extension of the classical result presented in Section XI 11. in \cite{MR1336382} and \cite{MR14095} (see also \cite{Tan}). This theorem  presents a duality between any Hausdorff compact group and a suitable subset of the ring of the Fourier polynomials associated with the group (see \cite{MR1336382} and \cite{MR14095} for details about that result). 

 We refer the reader to \cite{MR548112} and \cite{zbMATH03840637} for general results in Noncommutative integration (see also Section 5 in \cite{MR4378494} for a synthetic presentation of the topic). The characterizations of that kind of observables were well understood in \cite{zbMATH03840637} for the setting of measurable groupoids
 
 As mentioned in \cite{Ibort1}: 
 {\it Contrary to the situation with groups, even if G is locally compact there is not a canonical
(right/left) “invariant” measure on the groupoid, but a family of Haar measures had to be chosen.}

The paper is organized as follows:

In section \ref{topological-groupoids-section}, we introduce the concept of a topological groupoid. Next, we introduce characterizations of some observables defined in that setting, putting special attention on the example of ``the quantum ratchet'', which are presented explicit expressions. After that, we extend those characterizations to the context of non-finite compact groupoids.

In section \ref{tannaka-lemma-section} we present an extension of the so-called Tannaka's Theorem to the context of Hausdorff compact groupoids. First, we present the definitions of linear representation and unitary representation for topological groupoids (in a similar way to the ones presented in \cite{zbMATH07092992}). Later, we present the proof of the main result of the section which guarantees the existence of a suitable epimorphism between any Hausdorff compact groupoid and a subset of the ring of Fourier polynomials associated with the groupoid.

General results concerning groupoids in Quantum Mechanics can be found in \cite{Lat}, \cite{Ni} and  \cite{LCH}.

\section{Topological groupoids}
\label{topological-groupoids-section}

In this section, we are concerned with the study of compact topological groupoids, and some observables defined in that context, which  are widely studied in the settings of statistical mechanics and quantum mechanics. The section is divided into two parts. In the first one we deal with finite groupoids, in fact, we can induce a structure of $C^*$-algebra on the space of continuous functions defined on that class of groupoids and we present a suitable definition of quasi-invariant measure for that context. In the second part, we extend the results presented in the first one to the setting of compact groupoids using the so-called Haar systems.

So, let us first present the definition of {\it topological groupoid}. Consider a topological space $\mathcal{G}$ equipped with a {\it composition map} $(\alpha, \beta) \in \mathcal{G}^{(2)} \mapsto \alpha \cdot \beta \in \mathcal{G}$, where $\mathcal{G}^{(2)} \subset \mathcal{G} \times \mathcal{G}$ is the space of {\it composable pairs} (which is equipped with the subspace topology), we also consider a {\it inverse map} $\alpha \in \mathcal{G} \to \alpha^{-1} \in \mathcal{G}$ and we assume that the maps $\cdot$ and $^{-1}$ are continuous and satisfy the following properties:

\begin{enumerate}[i)]
\item $(\alpha^{-1})^{-1} = \alpha$ for each $\alpha \in \mathcal{G}$.
\item For any $(\alpha, \alpha') \in \mathcal{G}^{(2)}$ and $(\alpha', \alpha'') \in \mathcal{G}^{(2)}$, we have $(\alpha \cdot \alpha', \alpha'') \in \mathcal{G}^{(2)}$, $(\alpha, \alpha' \cdot \alpha'') \in \mathcal{G}^{(2)}$ and $(\alpha \cdot \alpha') \cdot \alpha'' = \alpha \cdot (\alpha' \cdot \alpha'')$. 
\item Each $\alpha \in \mathcal{G}$ satisfies $(\alpha^{-1}, \alpha) \in \mathcal{G}^{(2)}$, and for any $(\alpha, \beta) \in \mathcal{G}^{(2)}$, we have $\alpha^{-1} \cdot (\alpha \cdot \beta) = \beta$.
\item Any $\alpha \in \mathcal{G}$ satisfies $(\alpha, \alpha^{-1}) \in \mathcal{G}^{(2)}$, and for each $(\beta, \alpha) \in \mathcal{G}^{(2)}$, we have $(\beta \cdot \alpha) \cdot \alpha^{-1} = \beta$.
\end{enumerate}

We call any element $\alpha \in \mathcal{G}$ a {\it morphism} of the groupoid, we define the {\it source} of $\alpha$ as $s(\alpha) := \alpha^{-1} \cdot \alpha$ and the {\it range} of $\alpha$ as $r(\alpha) := \alpha \cdot \alpha^{-1}$. The space of {\it objects} of the groupoid $\mathcal{G}$ is defined as $\mathcal{G}^{(0)} := s(\mathcal{G}) = r(\mathcal{G}) \subset \mathcal{G}$ (which is equipped with the subspace topology). We call {\it units} of the groupoid $\mathcal{G}$ to the elements belonging to the space $\mathcal{G}^{(0)}$. The above, because each $\alpha \in \mathcal{G}$ satisfies $\alpha \cdot s(\alpha) = \alpha$ and $r(\alpha) \cdot \alpha = \alpha$. Moreover, since the maps $\cdot$ and $^{-1}$ are continuous, it follows that the maps $\alpha \in \mathcal{G} \mapsto r(\alpha) \in \mathcal{G}^{(0)}$, $\alpha \in \mathcal{G}  \mapsto s(\alpha) \in \mathcal{G}^{(0)}$ are continuous as well.

Given a pair $a, b \in \mathcal{G}^{(0)}$, we use the notations $\mathcal{G}_a := r^{-1}(\{a\})$, $\mathcal{G}^b := s^{-1}(\{b\})$ and $\mathcal{G}_a^b := s^{-1}(\{b\}) \cap r^{-1}(\{a\})$. Actually, note that two morphisms $\alpha \in \mathcal{G}_a^b$ and $\alpha' \in \mathcal{G}_{a'}^{b'}$ are {\it composable} in the form $\alpha \cdot \alpha'$ if, and only if, $b = a'$, i.e., when $s(\alpha) = r(\alpha')$. Furthermore, under the former assumptions we obtain that $\alpha \cdot \alpha' \in \mathcal{G}_a^{b'}$ and the groupoid $\mathcal{G}$ can be expressed as
$$
\mathcal{G} := \bigcup_{a, b \in \mathcal{G}^{(0)}} \mathcal{G}_a^b \;,
$$

\begin{definition} \label{top} Any topological space $\mathcal{G}$ satisfying the conditions mentioned above is called a {\it topological groupoid} and it is usual to use the notation $\mathcal{G} \rightrightarrows \mathcal{G}^{(0)}$.
\end{definition}

On the other hand, given an arbitrary $a \in \mathcal{G}^{(0)}$, the {\it isotropy group} for the object $a$ is defined as the space of all the morphisms belonging to the set $\mathcal{G}_a^a$ equipped with the composition law $\cdot$. Actually, it is easy to check that $\mathcal{G}_a^a$ is a group by the properties of $\mathcal{G}$.

Throughout the paper we use the notation $\mathcal{B}_\mathcal{G}$ for the collection of {\it Borel sets} on $\mathcal{G}$, we also denote by $\mathrm{C}(\mathcal{G})$ the set of {\it complex continuous functions} on $\mathcal{G}$ equipped with the uniform norm $\|\cdot\|_\infty$, we use the notation $\mathcal{M}(\mathcal{G})$ for the set of {\it finite Borel measures} on $\mathcal{G}$ and we denote by $\mathcal{M}_1(\mathcal{G})$ the set of {\it Borel probability measures} on $\mathcal{G}$.

\begin{remark}
Note that any topological group $G$ is a topological groupoid where the space of units is $G^{(0)} = \{{\bf 1}\}$ (with ${\bf 1}$ the unit of the topological group). Furthermore, in that case, $G$ has an only isotropy group and it is the trivial group. By the above, all the theories presented throughout this section can be applied to the framework of topological groups, in particular, the so-called isotropy groups as we will see in the following section.
\end{remark}

\subsection{Finite groupoids}
\label{finite-groupoids-section}

We say that a topological groupoid $\mathcal{G}$ is a {\it finite groupoid} when the set of morphisms is finite. Actually, in this case, the spaces $\mathcal{G}^{(0)}$, $\mathcal{G}$ and $\mathcal{G}^{(2)}$ are all finite and each one of them is equipped with the discrete topology. In particular, the above implies that any set in $\mathcal{B}_\mathcal{G}$ is a finite union of morphisms belonging to $\mathcal{G}$. Throughout this section, we deal only with finite groupoids.

Next, we will present an interesting example of a finite groupoid (beyond the group theory approach), well known in the mathematical physics literature as ``quantum ratchet". That model describes a physical system of {\it selective measures} acting on microscopic phenomena (see for more details see  \cite{Ibort2} and \cite{Ibort1}). In addition, we also describe in explicit form the main concepts usually considered in Noncommutative integration (see \cite{zbMATH03840637} and \cite{MR548112}), and we detail appropriate definitions and examples in that matter.

\begin{example}[Quantum ratchet]
\label{quantum-ratchet-definition}
Consider the set $\{-, +\}$ and the {\it permutation map} $\sigma : \{1,2,3\} \to \{1,2,3\}$, given by, $\sigma(1) = 2$, $\sigma(2) = 3$ and $\sigma(3) = 1$. {\it The quantum ratchet} groupoid $\mathcal{G}$ is given by the expression
\begin{equation} 
\label{pois} 
\mathcal{G} := \{(a, \sigma_0,  b), (a, \sigma_1, b), (a, \sigma_2, b ) :\; a, b \in \{-, +\} \} \;,
\end{equation}
where $\sigma_j := \sigma^j$ for $j = 0, 1, 2$. Following 
\cite{Ibort1} and \cite{Ibort2}, we denote the elements of $\mathcal{G}$ by
\begin{center}
\begin{tabular}{ c c c }
$+ := (+, \sigma_0, +)$, & $\sigma_+ := (+, \sigma_1, +)$, & $\sigma^2_+ := (+, \sigma_2, +)$, \\
$- := (-, \sigma_0, -)$, & $\sigma_- := (-, \sigma_1, -)$, & $\sigma^2_- := (-, \sigma_2, -)$, \\
$\beta_1 := (-, \sigma_1, +)$, & $\beta_2 := (-, \sigma_2, +)$, & $\beta_3 := (-, \sigma_0, +)$, \\
$\alpha_1 := (+, \sigma_1, -)$, & $\alpha_2 := (+, \sigma_2, -)$, & $\alpha_3 := (+, \sigma_0, -)$.
\end{tabular}
\end{center}

As at the beginning of the section, we denote by $\alpha$ the general element in the groupoid $\mathcal{G}$. We say that two morphisms $\beta = (a_1, \sigma_i, b_ 1)$ and $\alpha = (a_2, \sigma_j, b_ 2)$, with $i, j = 0, 1, 2$, are composable in the form $\alpha \cdot \beta$ if, and only if, $a_1 = b_2$. Furthermore, in this case, the composition law $\cdot$ is given by
\begin{equation}
\label{composition-law}
\alpha \cdot \beta 
= ( a_2 , \sigma_j, a_1 ) \cdot ( a_1 , \sigma_i, b_1 ) := ( a_2 , \sigma_j \sigma_i, b_1 ) \;,
\end{equation}
where $\sigma_j \sigma_i$ is the usual product of permutations on the set $\{1, 2, 3\}$.

On the other hand, since each one of the maps $\sigma_j$, with $j = 0, 1, 2$, is a bijection on $\{1,2,3\}$, by \eqref{composition-law}, it is guaranteed that each element of $\mathcal{G}$ has an inverse. The list of all the inverse morphisms is given by
\begin{center}
\begin{tabular}{ c c c c }
$+^{-1} = +$, & $\beta_1^{-1} = \alpha_2$, & $\alpha_1^{-1} = \beta_2$, & $-^{-1} = -$, \\
$\sigma_+^{-1} = \sigma^2_+$, & $\beta_2^{-1} = \alpha_1$, & $\alpha_2^{-1} = \beta_1$, & $\sigma_-^{-1} = \sigma^2_-$, \\
$(\sigma^2_+)^{-1} = \sigma_+$, & $\beta_3^{-1} = \alpha_3$, & $\alpha_3^{-1} = \beta_3$, & $(\sigma^2_-)^{-1} = \sigma_-$.
\end{tabular}
\end{center}

By the above, given a triple $\alpha = (a, \sigma_j, b )$ belonging to $\mathcal{G}$, with $j = 0, 1, 2$ and $a, b \in \{-, +\}$, the source of $\alpha$ is equal to $s(\alpha) = s(a, \sigma_j, b) = (b, \sigma_0, b) := b$ and the range of $\alpha$ is given by $r(\alpha)= r(a, \sigma_j, b) := (a, \sigma_0, a) := a$. Therefore, the set $\mathcal{G}$, such as appears defined in \eqref{pois}, joined with the composition law in \eqref{composition-law} and the inverse map results in a finite groupoid with a set of objects $\mathcal{G}^{(0)} = \{-, +\}$. 
\end{example}

\begin{remark}
Note that the groupoid $\mathcal{G}$ presented in the last example has two units which imply that $\mathcal{G}$ cannot be isomorphic to any group. 
\end{remark}

Given $\alpha \in \mathcal{G}$, we define the {\it canonical potential} $\chi_\alpha : \mathcal{G} \to \mathbb{C}$ by the equation
\begin{equation*}
\chi_\alpha(\beta) :=
\begin{cases}
1 &, \alpha = \beta \\
0 &, \alpha \neq \beta \;.
\end{cases}
\end{equation*}

It is not difficult to check that $\{\chi_\alpha :\; \alpha \in \mathcal{G}\}$ results in a basis for the space $\mathrm{C}(\mathcal{G})$. That is, any $f \in \mathrm{C}(\mathcal{G})$ can be expressed as 
\begin{equation}
\label{lcp}
f = \sum_{\alpha \in \mathcal{G}} z_\alpha \chi_\alpha \;,
\end{equation} 
with $z_\alpha \in \mathbb{C}$. Actually, by the expression in \eqref{lcp}, it is easy to realize that $z_\alpha := f(\alpha)$ for each $\alpha \in \mathcal{G}$. 

Furthermore, since $\mathcal{G}$ is a groupoid, we have that each $\alpha \in \mathcal{G}$ has an inverse $\alpha^{-1} \in \mathcal{G}$. Then, we are able to define an {\it adjoint operator} $\mathrm{Ad} : \mathrm{C}(\mathcal{G}) \to \mathrm{C}(\mathcal{G})$ assigning to each $f \in \mathrm{C}(\mathcal{G})$ the function $\mathrm{Ad}(f) = f^* \in \mathrm{C}(\mathcal{G})$ given by the equation
\begin{equation}
\label{adjoint-op}
f^* := \sum_{\alpha \in \mathcal{G}} \overline{f(\alpha)}\chi_{\alpha^{-1}} \;.
\end{equation}

By \eqref{adjoint-op}, it follows that $f^*(\alpha) = \overline{f(\alpha^{-1})}$ for each $\alpha \in \mathcal{G}$ and it is easy to check that $\mathrm{Ad}$ is an involution, i.e., $(f^*)^* = f$. Furthermore, below we will show that the operator $\mathrm{Ad}$ results in an involution for a suitable Banach algebra defined on the space $\mathrm{C}(\mathcal{G})$.

\begin{example}[Involution]
Consider the so-called ``quantum ratchet" (see Example \ref{quantum-ratchet-definition} for details). So, the adjoint operator can be expressed explicitly by
\begin{align*}
f^*
=& \overline{f(+)}\chi_+ + \overline{f(\sigma_+)}\chi_{\sigma_+^2} + \overline{f(\sigma_+^2)}\chi_{\sigma_+} + \overline{f(-)}\chi_- + \overline{f(\sigma_-)}\chi_{\sigma_-^2} + \overline{f(\sigma_-^2)}\chi_{\sigma_-} \nonumber \\
&+ \overline{f(\alpha_1)}\chi_{\beta_2} + \overline{f(\alpha_2)}\chi_{\beta_1} + \overline{f(\alpha_3)}\chi_{\beta_3} + \overline{f(\beta_1)}\chi_{\alpha_2} + \overline{f(\beta_2)}\chi_{\alpha_1} + \overline{f(\beta_3)}\chi_{\alpha_3} \nonumber \;.
\end{align*}

For instance, evaluating in $\alpha_1$, we obtain that $f^*(\alpha_1) = \overline{f(\beta_2)}$.
\end{example}

Now we introduce a {\it convolution operation} $*$ on the space $\mathrm{C}(\mathcal{G})$ in the following way: given two potentials $f, g \in \mathrm{C}(\mathcal{G})$, we set the convolution between $f$ and $g$ as
\[
f * g
:= \sum_{s(\alpha) = r(\beta)} f(\alpha)g(\beta)\chi_{\alpha \cdot \beta} \;.
\]

In fact, it is not difficult to check that the above expression implies that each $\gamma \in \mathcal{G}$ satisfies
\[
f * g(\gamma) = \sum_{\alpha \cdot \beta = \gamma} f(\alpha)g(\beta) = \sum_{r(\alpha) = r(\gamma)} f(\alpha)g(\alpha^{-1} \cdot \gamma) \;.
\]

Furthermore, it is not difficult to check that the convolution defined above satisfies $(f * g)^* = g^* * f^*$, where $^*$ is the involution in \eqref{adjoint-op}. Hereafter, we will use the notation $\sigma_+^0 := +$ and $\sigma_-^0 := -$ in order to simplify the expressions appearing in the examples presented below.

\begin{example}[Convolution]
The explicit expression for the convolution in the case of the ``quantum ratchet" (introduced in Example \ref{quantum-ratchet-definition}), is given by
\begin{align}
f * g
=& \Bigl( \sum_{j = 1}^3 f(\sigma^{j - 1}_+) \Bigr)\Bigl( \sum_{i = 1}^3 g(\beta_i)\chi_{\sigma^{j - 1}_+ \cdot \beta_i} + g(\sigma^{i - 1}_+)\chi_{\sigma^{j - 1}_+ \cdot \sigma^{i - 1}_+}\Bigr) \nonumber \\
&+ \Bigl( \sum_{j = 1}^3 f(\alpha_j) \Bigr)\Bigl( \sum_{i = 1}^3 g(\beta_i)\chi_{\alpha_j \cdot \beta_i} + g(\sigma^{i - 1}_+)\chi_{\alpha_j \cdot \sigma^{i - 1}_+}\Bigr) \nonumber \\
&+ \Bigl( \sum_{j = 1}^3 f(\sigma^{j - 1}_-) \Bigr)\Bigl( \sum_{i = 1}^3 g(\alpha_i)\chi_{\sigma^{j - 1}_- \cdot \alpha_i} + g(\sigma^{i - 1}_-)\chi_{\sigma^{j - 1}_- \cdot \sigma^{i - 1}_-}\Bigr) \nonumber \\
&+ \Bigl( \sum_{j = 1}^3 f(\beta_j) \Bigr)\Bigl( \sum_{i = 1}^3 g(\alpha_i)\chi_{\beta_j \cdot \alpha_i} + g(\sigma^{i - 1}_-)\chi_{\beta_j \cdot \sigma^{i - 1}_-}\Bigr) \label{rew} \;.
\end{align}

For instance, evaluating in the morphism $+ \in \mathcal{G}$, we obtain
\begin{align*}
f * g(+) 
=& f(\sigma_+)g(\sigma^2_+) + f(\sigma^2_+)g(\sigma_+) + f(+)g(+) \\
&+ f(\alpha_1)g(\beta_2) + f(\alpha_2)g(\beta_1) + f(\alpha_3)g(\beta_3) \;. 
\end{align*}
\end{example}

Observe that the space $\mathrm{C}(\mathcal{G})$ joint with the convolution $*$ becomes an algebra. Moreover, when $\mathrm{C}(\mathcal{G})$ is equipped with the uniform norm $\| \cdot \|_\infty$ and the involution in \eqref{adjoint-op}, we obtain a $C^*$-algebra which we denote by $\mathrm{C}^*(\mathcal{G})$. 

The same conclusion holds true on the {\it units space} $\mathcal{G}^{(0)}$, i.e., the space of functions $\mathrm{C}(\mathcal{G}^{(0)})$ equipped with the convolution $*$ and the involution in \eqref{adjoint-op} also results in a $C^*$-algebra which will be denoted by $\mathrm{C}^*(\mathcal{G}^{(0)})$. 

We say that a linear operator $\rho : \mathrm{C}^*(\mathcal{G}) \to \mathbb{C}$ is a {\it density operator} when $\rho(1) = 1$. Actually, in order to give an explicit expression, it is enough to identify the values $\rho(\chi_\alpha)$, for each $\alpha \in \mathcal{G}$, and guarantee that 
\begin{equation*}
\rho(1) = \sum_{\alpha \in \mathcal{G}} \rho(\chi_\alpha) = 1 \;.
\end{equation*}

The above, because by linearity of the space $\mathrm{C}(\mathcal{G})$, we have
\begin{equation}
\label{density-operator}
\rho(f) = \sum_{\alpha \in \mathcal{G}} f(\alpha) \rho(\chi_\alpha) \;.
\end{equation}

Density operators play an important role in the setting of Quantum Mechanics. The above is because they have the same behavior that the probability measures in the classical approach.

\begin{example}[Density operator]
In the case of the ``quantum ratchet" (which was introduced in Example \ref{quantum-ratchet-definition}), we have that a density operator defined on $\mathrm{C}^*(\mathcal{G})$ is given by
\begin{equation}
\label{dez}
\rho(f) := \sum_{i = 1}^3 f(\alpha_i)\rho(\chi_{\alpha_i}) + f(\beta_i)\rho(\chi_{\beta_i}) + f(\sigma^{i - 1}_+)\rho(\chi_{\sigma^{i - 1}_+}) + f(\sigma^{i - 1}_-)\rho(\chi_{\sigma^{i - 1}_-}) \;, 
\end{equation}
with 
\begin{equation*}
\sum_{i = 1}^3 \rho(\chi_{\alpha_i}) + \rho(\chi_{\beta_i}) + \rho(\chi_{\sigma^{i - 1}_+}) + \rho(\chi_{\sigma^{i - 1}_-}) = 1 \;. 
\end{equation*}
\end{example}

We say that a linear map $\lambda : \mathrm{C}^*(\mathcal{G}) \to \mathrm{C}^*(\mathcal{G}^{(0)})$ is fibered by the maps $r$ and $\mathrm{Id}_{\mathcal{G}^{(0)}}$ when for any $f \in \mathrm{C}(\mathcal{G})$ and each $a \in \mathcal{G}^{(0)}$ such that $f|_{\mathcal{G}_a} \equiv 0$, we have $\lambda(f)(a) = 0$.

\begin{definition} \label{GK}
The  {\it $\mathcal{G}$-kernels} are defined as linear maps $\lambda : \mathrm{C}^*(\mathcal{G}) \to \mathrm{C}^*(\mathcal{G}^{(0)})$ fibered by $r$ and $\mathrm{Id}_{\mathcal{G}^{(0)}}$ satisfying the expression
\begin{equation}
\label{G-kernel}
\lambda^a(f) = \lambda(f)(a) := \sum_{\alpha \in \mathcal{G}_a} t_\alpha^a f(\alpha) \;.
\end{equation}
\end{definition}

Since, any $E \in \mathcal{B}_\mathcal{G}$ is finite union of elements belonging to $\mathcal{G}$, denoting the {\it Dirac measure} by 
\begin{equation*}
\delta_\alpha(E) :=
\begin{cases}
1 &, \alpha \in E \\
0 &, \alpha \notin E \;.
\end{cases}
\end{equation*}
it follows that for each $a \in \mathcal{G}^{(0)}$
\begin{equation*}
\lambda^a = \lambda( \cdot )(a) = \sum_{\alpha \in \mathcal{G}_a} t_\alpha^a \delta_\alpha \;,
\end{equation*}

That is, the map $a \mapsto \lambda^a$ takes values into the set $\mathcal{M}(\mathcal{G})$. Moreover, by the fibered property, we have $t_\alpha^a = 0$ when $\alpha \notin \mathcal{G}_a$ which also implies that the Borel measure $\lambda^a$ is supported on $\mathcal{G}_a \subset \mathcal{G}$. Besides that, defining
\begin{equation*}
\lambda(E)( \cdot ) := \lambda(\chi_E)( \cdot ) = \sum_{\alpha \in \mathcal{G}_\cdot} t_\alpha^{( \cdot )} \chi_E(\alpha) \;,
\end{equation*}
for each $E \subset \mathcal{B}_\mathcal{G}$, it follows that the map $E \mapsto \lambda(E)( \cdot )$ is measurable and takes values into the set $\mathrm{C}(\mathcal{G})$.

\begin{example}[$\mathcal{G}$-kernel]
\label{G-kernel-quantum-ratchet}
In the particular case of the ``quantum ratchet" (see Example \ref{quantum-ratchet-definition} for details), the explicit expression of the so-called $\mathcal{G}$-kernel is given by
\begin{equation*}
\lambda(f)(+) = \sum_{i = 1}^3 t_{\alpha_i}^+ f(\alpha_i) + t_{\sigma^{i - 1}_+}^+ f(\sigma^{i - 1}_+) \;,
\end{equation*} 
and
\begin{equation*}
\lambda(f)(-) = \sum_{i = 1}^3 t_{\beta_i}^- f(\beta_i) + t_{\sigma^{i - 1}_-}^- f(\sigma^{i - 1}_-) \;.
\end{equation*}
\end{example}

Next, we will introduce the definitions of transverse function, modular function, and quasi-invariant measure (see Definition \ref{quasi}) in our setting. These concepts  are the building block for the framework of noncommutative integration (see \cite{zbMATH03840637} and \cite{MR548112}).  In \cite{MR4378494} it is described the relation of Haar measures with DLR probabilities of Statistical Mechanics (which can be seen in the framework of Thermodynamic Formalism as eigenprobabilities for the dual of the Ruelle operator). On the other hand, the study of groupoids with dynamical content appears in  \cite{zbMATH07352601}, \cite{MR3993188}, \cite{MR2536186}, \cite{MR2215776}.

\begin{definition} \label{tran}
A {\it transverse function} is a $\mathcal{G}$-kernel satisfying the equation $\lambda^{r(\alpha)} = \alpha \lambda^{s(\alpha)}$ for each $\alpha \in \mathcal{G}$. Then, by the expression \eqref{G-kernel}, it follows that 
\begin{equation*}
\sum_{\gamma \in \mathcal{G}_{r(\alpha)}} t_\gamma^{r(\alpha)} \delta_\gamma = \sum_{\gamma \in \mathcal{G}_{s(\alpha)}} t_\gamma^{s(\alpha)} \delta_{\alpha \cdot \gamma} \;,
\end{equation*}
for each $\alpha \in \mathcal{G}$. The above is equivalent to say
\begin{equation*}
\sum_{\gamma \in \mathcal{G}} t_\gamma^{r(\alpha)} f(\gamma) = \sum_{\gamma \in \mathcal{G}} t_\gamma^{s(\alpha)} f(\alpha \cdot \gamma) \;,
\end{equation*}
for each $\alpha \in \mathcal{G}$ and any $f \in \mathrm{C}^*(\mathcal{G})$. In particular, when we have $\mathcal{G}_a \neq \emptyset$ for each $a \in \mathcal{G}^{(0)}$, we say that $\lambda$ is a {\it faithful transverse function}. 
\end{definition}


\begin{definition} \label{mod}

We say that a map $\Delta : \mathcal{G}^{(2)} \to \mathbb{R}^+$ is a {\it modular form} when each $(\alpha^{-1}, \beta) \in \mathcal{G}^{(2)}$ satisfies
\begin{equation*}
\Delta(\alpha^{-1} \cdot \beta) = \Delta(\alpha)^{-1}\Delta(\beta) \;.
\end{equation*}

\end{definition}

More general definitions will be presented later.

Using the above, we can present a suitable definition of quasi-invariant measure with respect to a modular function $\Delta : \mathcal{G} \to \mathbb{C}$ in the matter of finite groupoids. Indeed, consider a faithful transverse function $\lambda$ and $M \in \mathcal{M}(\mathcal{G}^{(0)})$. It is not difficult to check that the measure $M$ is of the form
\begin{equation*}
M := \sum_{a \in \mathcal{G}^{(0)}} t_a \delta_a \;,
\end{equation*}
where the Dirac measure at the point $a \in \mathcal{G}^{(0)}$ is given by
\begin{equation*}
\delta_a(E) :=
\begin{cases}
1 &, a \in E \\
0 &, a \notin E \;,
\end{cases}
\end{equation*} 
for each $E \in \mathcal{B}_{\mathcal{G}^{(0)}}$. We say that the measure $M$ is a {\it quasi-invariant measure} (general definition in Definition  \ref{quasi}) with respect to the modular function $\Delta$, if we have that each $f \in \mathrm{C}^*(\mathcal{G})$ satisfies the expression
\begin{equation}
\label{quasi-inv}
\sum_{a \in \mathcal{G}^{(0)}}\sum_{\alpha \in \mathcal{G}_a} t_a t_\alpha^a f(\alpha) = \sum_{a \in \mathcal{G}^{(0)}}\sum_{\alpha \in \mathcal{G}_a} t_a t_{\alpha^{-1}}^a \overline{f(\alpha^{-1})} \Delta(\alpha^{-1}) \;.
\end{equation}

In particular, when the modular function is given by $\Delta(\alpha) = e^{\varphi(s(\alpha)) - \varphi(r(\alpha))}$, where $\varphi : \mathcal{G}^{(0)} \to \mathbb{C}$ is a measurable function, we say that $M$ is a {\it Haar-invariant measure} and the expression in \eqref{quasi-inv} is equivalent to
\begin{equation*}
\sum_{a \in \mathcal{G}^{(0)}}\sum_{\alpha \in \mathcal{G}_a} t_a t_\alpha^a f(\alpha) = \sum_{a \in \mathcal{G}^{(0)}}\sum_{\alpha \in \mathcal{G}_a} t_a t_{\alpha^{-1}}^a \overline{f(\alpha^{-1})} \frac{e^{\varphi(a)}}{e^{\varphi(s(\alpha))}} \;.
\end{equation*}

Moreover, note that in this case the measure $m := M \cdot \lambda$ satisfies the expression
\begin{equation*}
m = \sum_{a \in \mathcal{G}^{(0)}}\sum_{\alpha \in \mathcal{G}_a} t_a t_\alpha^a \delta_a \delta_\alpha \;,
\end{equation*}
and the property $m(f) = m(f^*\Delta^{-1})$.

\begin{example}
Now we present the explicit expression of the so-called Haar-invariant measures in the setting of the ``quantum ratchet" (introduced in the Example \ref{quantum-ratchet-definition}). In this case, we have a measure of the form
\begin{equation*}
M = t_- \chi_- + t_+ \chi_+ \;,
\end{equation*}
is a Haar-invariant measure with respect to the $\mathcal{G}$-kernel in Example \ref{G-kernel-quantum-ratchet} and the modular function $\Delta(\alpha) = e^{\varphi(s(\alpha))- \varphi(r(\alpha))}$, when satisfies the following 
\begin{align*}
&\sum_{i = 1}^3 t_+ (t_{\alpha_i}^+ f(\alpha_i) + t_{\sigma^{i - 1}_+}^+ f(\sigma^{i - 1}_+)) + t_-(t_{\beta_i}^- f(\beta_i) + t_{\sigma^{i - 1}_-}^- f(\sigma^{i - 1}_-)) \\
&= \sum_{i = 1}^3 t_+ \Bigl(t_{\beta_i}^+ \overline{f(\beta_i)}\frac{e^{\varphi(-)}}{e^{\varphi(+)}} + t_{\sigma^{i - 1}_+}^+ \overline{f(\sigma^{i - 1}_+)}\Bigr) + t_- \Bigl(t_{\alpha_i}^- \overline{f(\alpha_i)}\frac{e^{\varphi(+)}}{e^{\varphi(-)}} + t_{\sigma^{i - 1}_-}^- \overline{f(\sigma^{i - 1}_-)}\Bigr) \;,
\end{align*}
for each $f \in \mathrm{C}^*(\mathcal{G})$.
\end{example}

\subsection{Compact groupoids}
\label{non-finite-groupoids-section}

In this section, we present examples of compact non-finite topological groupoids. The main idea is to show some illustrative cases where the results to be presented in section \ref{tannaka-lemma-section} are applicable. Besides that, we present explicit expressions of some observables defined in section \ref{finite-groupoids-section} for the setting introduced in this section.

A more detailed analysis of properties related to  the next example as some examples of $C^*$-algebras in that context appear in \cite{MR2536186}, \cite{Ren1}, \cite{Put} and \cite{MR2215776}.

\begin{example}[Dynamic groupoid]
Consider a compact metric space $X$ and a homeomorphism $f : X \to X$. We assume that the group $G := \mathbb{Z}$ acts on the set $X$ taking any pair $(x, n) \in X \times \mathbb{Z}$ and sending it into the value $x \cdot n = f^n(x) \in X$. In this case, we can define the following groupoid
\[
\mathcal{G} := \{(x, n) :\; x \in X, \; n \in \mathbb{Z}\} \cong X \times \mathbb{Z} \;,
\]
where the set of objects of $\mathcal{G}$ is given by 
\[
\mathcal{G}^{(0)} := \{(x, 0) :\; x \in X\} = X \times \{0\} \cong X \;,
\] 
and the set of composable elements of $\mathcal{G}$ is given by 
\[
\mathcal{G}^{(2)} := \{((x, n), (x', m)) :\; f^n(x) = x'\} \;.
\]

In this case, we have that the product on $\mathcal{G}$ is defined as
\[
(x, n) \cdot (x', m) := (x, n+m) \;, 
\]
and any element $(x, n) \in \mathcal{G}$ has inverse $(x, n)^{-1} := (f^n(x), -n)$. Indeed,
\[
(x, n) \cdot (x, n)^{-1} := (x, n) \cdot (f^n(x), -n) = (x, 0) \in \mathcal{G}^{(0)} \;,
\]
and 
\[
(x, n)^{-1} \cdot (x, n) := (f^n(x), -n) \cdot (x, n) = (f^n(x), 0) \in \mathcal{G}^{(0)} \;,
\]

It is not difficult to check that both of the maps $^{-1} : \mathcal{G} \to \mathcal{G}$ and $\cdot : \mathcal{G}^{(2)} \to \mathcal{G}$ are continuous. The above, simply as a consequence of the continuity of the maps $\alpha \mapsto f(\alpha)$, $n \mapsto -n$ and $(n, m) \mapsto n+m$.

Furthermore, in this case, the source of the map $s$ and range $r$ are given by
$$
s(x, n) := (x, n)^{-1} \cdot (x, n) = (f^n(x), 0) \in \mathcal{G}^{(0)} \;,
$$ 
and 
$$
r(x, n) := (x, n) \cdot (x, n)^{-1} = (x, 0) \in \mathcal{G}^{(0)} \;,
$$
which guarantees continuity because the composition of continuous maps also results in a continuous map.

On the other hand, the so-called modular forms are  of the form 
\[
\Delta(x, n+m) = \Delta(x, n)\Delta(f^n(x), n+m) \;.
\]

That is, any modular form defined in this setting results in a cocycle.
\end{example}

Observe that in the last example, we present a topological groupoid with a non-numerable set of objects and a non-numerable set of morphisms which is Hausdorff compact. In fact, that is the setting in which we are interested in this section. So, the definitions of the observables that appear previously need some modifications in order to still hold in this new setting.

The definition of the adjoint operator is given also by the map $\mathrm{Ad} : \mathrm{C}(\mathcal{G}) \to \mathrm{C}(\mathcal{G})$ such that $\mathrm{Ad}(f) = f^*$ satisfies the expression $f^*(\alpha) = \overline{f(\alpha^{-1})}$ for each $\alpha \in \mathcal{G}$. So the map $f \mapsto f^*$ results in an involution. The definitions of convolution, density operator, $\mathcal{G}$-kernel, transverse function, modular form, and quasi-invariant measure considered here are given in a similar way to the ones presented in \cite{zbMATH03840637}.

Indeed, remember that a linear map $\lambda : \mathrm{C}(\mathcal{G}) \to \mathrm{C}(\mathcal{G}^{(0)})$ is fibered by the maps $r$ and $\mathrm{Id}_{\mathcal{G}^{(0)}}$ when for any $f \in \mathrm{C}(\mathcal{G})$ and each $a \in \mathcal{G}^{(0)}$ such that $f|_{\mathcal{G}_a} \equiv 0$, we have $\lambda(f)(a) = 0$.

\begin{definition} \label{Ge} 
Given a collection $\{\lambda^a : a \in \mathcal{G}^{(0)}\}$, with $\lambda^a \in \mathcal{M}(\mathcal{G})$, define 
\begin{equation}
\label{G-kernel-con}
\lambda(f)(a) := \int_{\alpha \in \mathcal{G}} f(\alpha) d\lambda^a(\alpha) \;.
\end{equation}

Then, we say that $\lambda$ is a {\it $\mathcal{G}$-kernel} when it satisfies the following conditions
\begin{enumerate}[i)]
\item $\mathrm{Supp}(\lambda^a) = \mathcal{G}_a$ for each $a \in \mathcal{G}^{(0)}$;
\item For each $f \in \mathrm{C}(\mathcal{G})$, the map $a  \mapsto \lambda(f)(a)$ is continuous; 
\item We have $\int_{\mathcal{G}} f(\gamma) d\lambda^{r(\alpha)}(\gamma) = \int_{\mathcal{G}} f(\alpha \cdot \gamma) d\lambda^{s(\alpha)}(\gamma)$ for any $\alpha \in \mathcal{G}$ and each $f \in \mathrm{C}(\mathcal{G})$. 
\end{enumerate}

\end{definition}

It is not difficult to check that the former conditions imply that $\lambda$ is fibered by $r$ and $\mathrm{Id}_{\mathcal{G}^{(0)}}$. 

\begin{definition} \label{Ha}  When all of the conditions presented above are satisfied, we say that the collection of measures $\{\lambda^a : a \in \mathcal{G}^{(0)}\}$ is a {\it left Haar system} (see for instance \cite{MR2536186}).
\end{definition} 

Now we are able to induce a structure of $C^*$-algebra on the space of functions $\mathrm{C}(\mathcal{G})$, even when the groupoid is not a finite one. So, consider a left Haar system $\{\lambda^a : a \in \mathcal{G}^{(0)}\}$, we define the convolution $*_{\lambda}$ between two maps $f, g$ belonging to $\mathrm{C}(\mathcal{G})$ by the expression
\begin{equation*}
f *_\lambda g := \int_{\mathcal{G}} f(\alpha \cdot \gamma) g(\gamma^{-1}) d\lambda^{s(\alpha)}(\gamma) \;.
\end{equation*}

It follows immediately that the space $\mathrm{C}(\mathcal{G})$ equipped with the convolution $*_\lambda$ is an algebra. Moreover, when it is equipped with the involution $f \mapsto f^*$ and the norm $\| \cdot \|_\infty$, it results in a $C^*$-algebra which we denote again by $\mathrm{C}^*(\mathcal{G})$.

In this case   a $\mathcal{G}$-kernel is a {\it transverse function} when it satisfies the expression $\lambda^{r(\alpha)} = \alpha \lambda^{s(\alpha)}$ for each $\alpha \in \mathcal{G}$. That is, 
\begin{equation*}
\int_{\mathcal{G}}f(\gamma) d\lambda^{r(\alpha)}(\gamma) = \int_{\mathcal{G}} f(\alpha \cdot \gamma) d\lambda^{s(\alpha)}(\gamma) \;.
\end{equation*}

In particular (see \cite{zbMATH03840637} for details), the former condition is equivalent to say that for any $f, g, h \in \mathrm{C}(\mathcal{G})$ the following expression holds 
\begin{equation*}
(f *_\lambda g) *_\lambda h = f *_\lambda (g *_\lambda h) \;.
\end{equation*}

We say that $\lambda$ is a {\it faithful transverse function}, when $\mathcal{G}_a \neq \emptyset$ for each $a \in \mathcal{G}^{(0)}$.

Next, we want to present a suitable definition of quasi-invariant measure in this context. 

\begin{definition} \label{quasi} Given a {\it modular function} $\Delta : \mathcal{G}^{(2)} \to \mathbb{R}^+$, i.e., satisfying $\Delta(\alpha^{-1} \cdot \beta) = \Delta(\alpha)^{-1} \Delta(\beta)$, we say that $M \in \mathcal{M}(\mathcal{G}^{(0)})$ is a {\it quasi-invariant measure} with respect to $\Delta$, when each $f \in \mathrm{C}^*(\mathcal{G})$ satisfies
\begin{equation*}
\int_{\mathcal{G}^{(0)}} \Bigl( \int_{\mathcal{G}_a} f(\alpha) d\lambda^a(\alpha) \Bigr) dM(a) = \int_{\mathcal{G}^{(0)}} \Bigl( \int_{\mathcal{G}_a} \overline{f(\alpha^{-1})}\Delta(\alpha^{-1})d\lambda^a(\alpha) \Bigr) dM(a) \;.
\end{equation*}

That is, the measure $m := M \circ \lambda$ which  satisfies $m(f) = m(f^* \Delta^{-1})$.

\end{definition}

\section{Tannaka's Theorem}
\label{tannaka-lemma-section}

In this section,  we prove a version of the so-called Tannaka's Theorem in the setting of Hausdorff compact topological groupoids. That is, we want to prove a duality between the groupoid $\mathcal{G}$ and a suitable set of homomorphisms of algebras. 

Throughout this section we use the notation $\mathrm{C}^{\geq 0}(\mathcal{G})$ for the cone of {\it non-negative continuous functions} on $\mathcal{G}$. A {\it linear representation} of the groupoid $\mathcal{G}$ is defined as a functor $\Lambda : \mathcal{G} \to \mathrm{\bf FinVect}$, where $\mathrm{\bf FinVect}$ is the {\it category of finite dimensional vector spaces}, sending each $x \in \mathcal{G}^{(0)}$ into a finite dimensional vector space $\Lambda(x) := V_x$ and any $\alpha \in \mathcal{G}$ into a linear map $\Lambda(\alpha) : V_{r(\alpha)} \to V_{s(\alpha)}$ such that $\Lambda(\alpha \cdot \beta) = \Lambda(\alpha)\Lambda(\beta)$ for any pair $(\alpha, \beta) \in \mathcal{G}^{(2)}$ and $\Lambda(x) = \mathrm{Id}_{V_x}$ for each $x \in \mathcal{G}^{(0)}$. 

Since each $\alpha \in \mathcal{G}$ is an isomorphism with inverse $\alpha^{-1}$, it follows that $\Lambda(\alpha)$ is a linear isomorphism with inverse $\Lambda(\alpha)^{-1} = \Lambda(\alpha^{-1})$. Moreover, we have
\begin{align*}
\mathrm{Id}_{V_{r(\alpha)}} &= \Lambda(r(\alpha)) = \Lambda(\alpha \cdot \alpha^{-1}) = \Lambda(\alpha) \Lambda(\alpha^{-1}) \;, \\
\mathrm{Id}_{V_{s(\alpha)}} &= \Lambda(s(\alpha)) = \Lambda(\alpha^{-1} \cdot \alpha) = \Lambda(\alpha^{-1}) \Lambda(\alpha) \;. 
\end{align*}

\begin{remark}
It is well known that the spaces $V_x$ associated with a linear representation $\Lambda$ of $\mathcal{G}$ result isomorphic on the connected components of $\mathcal{G}$. In particular, when $\mathcal{G}$ is connected (i.e., when for any $x, y \in \mathcal{G}^{(0)}$ there is $\alpha \in \mathcal{G}$ such that $s(\alpha) = y$ and $r(\alpha) = x$), it follows that the linear spaces $V_x$ are all isomorphic.
\end{remark}

Given $\Lambda$ a linear representation of $\mathcal{G}$ and a connected component $\mathcal{G}_i$ of $\mathcal{G}$, we call {\it local dimension} of $\Lambda$ into the component $\mathcal{G}_i$, to the dimension of the spaces $\Lambda(x)$ (which is the same for each $x \in \mathcal{G}_i^{(0)}$). So, in the case of connected groupoids, the {\it dimension} of $\Lambda$ is defined as the dimension of $\Lambda(x)$, where $x$ is an arbitrary element of $\mathcal{G}^{(0)}$.

Our next goal is to introduce suitable definitions for direct sum and tensor product between linear representations of a topological groupoid $\mathcal{G}$.

So, the {\it direct sum} between linear representations is defined as a functor $\oplus : \mathrm{\bf FinVect} \times \mathrm{\bf FinVect} \to \mathrm{\bf FinVect}$, such that each $x \in \mathcal{G}^{(0)}$ satisfies
\begin{equation*}
(\Lambda \oplus \Lambda')(x) := \Lambda(x) \oplus \Lambda'(x)
\end{equation*}
and 
\begin{equation*}
(\Lambda \oplus \Lambda')(\alpha) := \Lambda(\alpha) \oplus \Lambda'(\alpha)
\end{equation*}
for each $\alpha \in \mathcal{G}$ (here the direct sum between linear transformations is given in the usual sense). Hereafter, we will use the notation $\Lambda \oplus \Lambda'$ for the direct sum of the linear representations $\Lambda$ and $\Lambda'$ which also results in a linear representation. 

On the other hand, the {\it tensor product} between linear representations is defined as a functor $\otimes : \mathrm{\bf FinVect} \times \mathrm{\bf FinVect} \to \mathrm{\bf FinVect}$ which satisfies 
\begin{equation*}
(\Lambda \otimes \Lambda')(x) = \Lambda(x) \otimes \Lambda'(x)
\end{equation*}
for each $x \in \mathcal{G}^{(0)}$ and 
\begin{equation*}
(\Lambda \otimes \Lambda')(\alpha) = \Lambda(\alpha) \otimes \Lambda'(\alpha)
\end{equation*}
for each $\alpha \in \mathcal{G}$ (where the tensor product between linear transformations is defined as usual). Note that the tensor product between the linear representations $\Lambda$ and $\Lambda'$ also results in a linear representation which we denote by $\Lambda \otimes \Lambda'$. 

Hereafter, we will use the notation $\mathrm{U}(V)$ for the set of {\it unitary transformations} from $V$ into itself, where $V$ is a {\it finite dimensional vector space} equipped with an {\it inner product} $( \cdot , \cdot )$.

An important characteristic that we want to identify among the linear representations of a groupoid $\mathcal{G}$ is when they are comparable. So, we say that $\varphi$ is a {\it morphism} between the linear representations $\Lambda$ and $\Lambda'$, when there is a family of linear transformations $(\varphi_x)_{x \in \mathcal{G}^{(0)}}$, with $\varphi_x : \Lambda(x) \to \Lambda'(x)$, such that
\begin{equation*}
\varphi_{s(\alpha)} \Lambda(\alpha) = \Lambda'(\alpha) \varphi_{r(\alpha)} \;.
\end{equation*}

Furthermore, we say that $\Lambda$ and $\Lambda'$ are {\it equivalent}, when  each one of the maps $\varphi_x $ is a linear isomorphism and, in that case, we use the notation $\Lambda \sim \Lambda'$. In particular, the spaces $\Lambda(x)$ and $\Lambda'(x)$ have the same dimension for each $x \in \mathcal{G}^{(0)}$ when $\Lambda \sim \Lambda'$.

The former definitions allow us to characterize the irreducibility of a linear representation. In fact, we say that the linear representation $\Lambda$ of $\mathcal{G}$ is {\it reducible}, when there are linear representations $\Lambda_1$ and $\Lambda_2$ of $\mathcal{G}$ such that
\begin{equation*}
\Lambda \sim \Lambda_1 \oplus \Lambda_2 \;,
\end{equation*}
in another case, we say the linear representation $\Lambda$ is {\it irreducible}.

Now we are able to present a definition of unitary representation in the setting of topological groupoids. Hereafter, we assume that $\mathcal{G}$ is a connected groupoid. So, given any linear representation $\Lambda$, we have $\Lambda(x) = V_\Lambda$ for each $x \in \mathcal{G}^{(0)}$. By the above, the vector space $V_\Lambda$ is uniquely determined by the functor $\Lambda$ and it is equipped with an {\it inner product} $( \cdot , \cdot )$. 

We say that $\Lambda$ is a {\it unitary representation} of the connected groupoid $\mathcal{G}$, when $\Lambda(\alpha) \in \mathrm{U}(V_\Lambda)$ for any $\alpha \in \mathcal{G}$. That is, any $\alpha \in \mathcal{G}$ satisfies $(\Lambda(\alpha)u, \Lambda(\alpha)v) = (u, v)$ for each pair $u, v \in V_\Lambda$. Furthermore, assuming that $n = \mathrm{dim}(V_\Lambda)$ and $\{e_1, ... , e_n\}$ is a orthonormal basis of $V_\Lambda$, it is possible to find a {\it matrix representation} of $\Lambda(\alpha)$ belonging to $\mathrm{U}(n, \mathbb{C})$ which we denote by $(\Lambda_{i j}(\alpha))_{n \times n}$. 

\begin{remark}
For instance, any unitary representation $\Lambda$ such that $\mathrm{dim}(V_\Lambda) = 1$ is a multiplication by a scalar $z$ such that $|z| = 1$. 
\end{remark}

A particular case of unitary representations is the so called {\it identity representations} denoted by $\Lambda_0$, which are defined as the functor sending each $\alpha \in \mathcal{G}$ into the map $\Lambda_0(\alpha) := \mathrm{Id}_{V_{\Lambda_0}}$. Note that, up to an isomorphism, there is a unique identity representation for each $n$-dimensional vector space.

Let $\mathcal{A}(\mathcal{G})$ be the set of {\it unitary representations} of the groupoid $\mathcal{G}$. Then, we have that $\Lambda \sim \Lambda'$ when there exists a unitary isomorphism $\varphi : V_\Lambda \to V_{\Lambda'}$, such that $\varphi \Lambda(\alpha) = \Lambda'(\alpha) \varphi$ for each $\alpha \in \mathcal{G}$. On the other hand, given any set $S \subset V_\Lambda$, we say that $S$ is stable by $\Lambda$ when $w \in S$ implies $\Lambda(\alpha)w \in S$ for any $\alpha \in \mathcal{G}$. In the following lemma, we show that stable vector subspaces characterize irreducible components of unitary representations.

\begin{lemma}
\label{reducible-lemma}
Consider $\Lambda \in \mathcal{A}(\mathcal{G})$ and suppose that there is a non-trivial subspace $W \subsetneq V_\Lambda$ stable by $\Lambda$. Then, $\Lambda$ is reducible.
\end{lemma}
\begin{proof}
Note that in this case $W^\bot$ is also stable by $\Lambda$. Indeed, taking any $v \in W^\bot$, i.e., a vector satisfying $(v, w) = 0$ for any $w \in W$, it follows that $(\Lambda(\alpha)v, w) = (v, \Lambda(\alpha^{-1})w) = 0$ for any $\alpha \in \mathcal{G}$. 

The above implies that $\Lambda(\alpha)v \in W^\bot$ for any $\alpha \in \mathcal{G}$. So, we have
\begin{equation*}
\Lambda \sim \Lambda|_W \oplus \Lambda|_{W^\bot} \;.
\end{equation*}
\end{proof}

\begin{remark}
The former lemma implies that any linear representation $\Lambda$, with $\mathrm{dim}(V_\Lambda) = n$, has a decomposition of the form $\Lambda = \Lambda_1 \oplus ... \oplus \Lambda_n$, where each $\Lambda_i$, with $i = 1, ... , n$ is irreducible.
\end{remark}

Given a groupoid $\mathcal{G}$ and a finite-dimensional vector space $V$, we denote by $\mathcal{A}(\mathcal{G}, V)$ the set of {\it unitary representations} of $\mathcal{G}$ on the space $V$. Fix $a \in \mathcal{G}^{(0)}$ and let $\mathrm{C}(a, \mathcal{G}^{(0)}, \mathrm{U}(V))$ be the set of {\it continuous functions} from $\mathcal{G}^{(0)}$ into the space $\mathrm{U}(V)$ sending $a$ into $\mathrm{Id}_V$. 

Since $\mathcal{G}$ is a connected topological groupoid, for each $a \in \mathcal{G}^{(0)}$ there exists a map $\Omega : \mathcal{G}^{(0)} \to \mathcal{G}$ such that each $x \in \mathcal{G}^{(0)}$ satisfies that $\Omega(x) \in \mathcal{G}_a^x$. For instance, for each $x \in \mathcal{G}^{(0)}$ we choose $\Omega(x) := \alpha^x$, for some $\alpha^x \in \mathcal{G}$ such that $s(\alpha^x) = x$ and $r(\alpha^x) = a$. However, in general, it is not possible to guarantee the continuity of the map $\Omega$. Because of that, we introduce the following hypothesis.

\begin{hypothesis}
Let $\mathcal{G}$ be a topological groupoid, we assume existence of $a \in \mathcal{G}^{(0)}$ and a continuous map $\Omega : \mathcal{G}^{(0)} \to \mathcal{G}$ such that each $x \in \mathcal{G}^{(0)}$ satisfies $\Omega(x) \in \mathcal{G}_a^x$. Furthermore, replacing $\Omega(x)$ by $\Omega(a)^{-1}\Omega(x)$, we can assume w.l.o.g. $\Omega(a) = {\bf 1}_a$.
\end{hypothesis}

\begin{lemma}
\label{lemma-H}
Let $\mathcal{G}$ be a connected compact groupoid satisfying the hypothesis {\bf (H)} for $a \in \mathcal{G}^{(0)}$ and $\Omega : \mathcal{G}^{(0)} \to \mathcal{G}$. Consider a finite dimensional vector space $V$ and $\Lambda \in \mathcal{A}(\mathcal{G}, V)$, we define $\lambda := \Lambda|_{\mathcal{G}_a^a}$ and $\mu(x) := \Lambda(\Omega(x))$ for each $x \in \mathcal{G}^{(0)}$. Then, the map sending each $\Lambda \in \mathcal{A}(\mathcal{G}, V)$ into $(\lambda, \mu) \in \mathcal{A}(\mathcal{G}_a^a, V) \times \mathrm{C}(a, \mathcal{G}^{(0)}, \mathrm{U}(V))$ results in a bijection.
\end{lemma}
\begin{proof}
Given $\alpha \in \mathcal{G}$, define the map $\Gamma : \mathcal{G} \to \mathcal{G}_a^a$ by
\[
\Gamma(\alpha) := \Omega(r(\alpha)) \cdot \alpha \cdot \Omega(s(\alpha))^{-1} \;.
\]

By the above, it follows that $\Gamma$ is a continuous homomorphism from $\mathcal{G}$ into $\mathcal{G}_a^a$. Moreover, we are able to obtain an expression of $\alpha$ in terms of the maps $\Gamma$, $\Omega$, $s$ and $r$ in the following way
\[
\alpha = \Omega(r(\alpha))^{-1} \cdot \Gamma(\alpha) \cdot \Omega(s(\alpha)) \;.
\]

Then, we have
\begin{align*}
\Lambda(\alpha) 
&= \Lambda(\Omega(r(\alpha))^{-1} \cdot \Gamma(\alpha) \cdot \Omega(s(\alpha))) \\
&= \Lambda(\Omega(r(\alpha)))^{-1}\Lambda(\Gamma(\alpha))\Lambda(\Omega(s(\alpha))) 
= \mu(r(\alpha))^{-1} \lambda(\Gamma(\alpha)) \mu(s(\alpha)) \;.
\end{align*}
which guarantees that the map $\Lambda \mapsto (\lambda, \mu)$ is a continuous bijection. 
\end{proof}

In the following lemma, we prove that the bijection above is independent of the $a \in \mathcal{G}^{(0)}$ given by the hypothesis {\bf (H)}.

\begin{lemma}
\label{lemma-H2}
Let $\mathcal{G}$ be a connected compact groupoid satisfying the hypothesis {\bf (H)} for $a \in \mathcal{G}^{(0)}$ and $\Omega : \mathcal{G}^{(0)} \to \mathcal{G}$. Then, given $x \in \mathcal{G}^{(0)}$ and $\lambda' \in \mathcal{A}(\mathcal{G}_x^x)$, there is $\Lambda' \in \mathcal{A}(\mathcal{G})$ such that $V_{\lambda'} = V_{\Lambda'}$ and $\lambda' = \Lambda'|_{\mathcal{G}_x^x}$.
\end{lemma}
\begin{proof}
Since $\mathcal{G}$ satisfies {\bf (H)} for $a \in \mathcal{G}^{(0)}$ and $\Omega : \mathcal{G}^{(0)} \to \mathcal{G}$, by Lemma \ref{lemma-H}, for any unitary representation $\lambda$ of $\mathcal{G}_a^a$, there is a unitary representation $\Lambda$ of $\mathcal{G}$ such that $V_\lambda = V_\Lambda$ and $\Lambda|_{\mathcal{G}_a^a} = \lambda$. 

Given $x \in \mathcal{G}^{(0)}$ and $\omega \in \mathcal{G}_a^x$, we define $\lambda(\alpha) := \lambda'(\omega^{-1} \cdot \alpha \cdot \omega)$ for each $\alpha \in \mathcal{G}_a^a$. It follows that $\lambda$ is a unitary representation of $\mathcal{G}_a^a$ and, therefore, there is $\Lambda$ unitary representation of $\mathcal{G}$ such that $V_\lambda = V_\Lambda$ and $\Lambda|_{\mathcal{G}_a^a} = \lambda$. Define $\Lambda'(x) := V_\Lambda$ and for each $\beta \in \mathcal{G}_x^x$ consider 
\begin{equation*}
\Lambda'(\beta) := \Lambda(\omega) \Lambda(\beta) \Lambda(\omega)^{-1} \;.
\end{equation*}

Since $x \in \mathcal{G}^{(0)}$ is arbitrary, it follows that $\Lambda'$ is a unitary representation of $\mathcal{G}$ (because $\Lambda$ is unitary representation of $\mathcal{G}$). Moreover, as $\omega \cdot \beta \cdot \omega^{-1} \in \mathcal{G}_a^a$, we obtain that
\begin{align*}
\Lambda'(\beta) 
= \Lambda(\omega) \Lambda(\beta) \Lambda(\omega)^{-1} 
&= \Lambda(\omega \cdot \beta \cdot \omega^{-1}) \\
&= \lambda(\omega \cdot \beta \cdot \omega^{-1}) 
= \lambda(\omega^{-1} \cdot (\omega \cdot \beta \cdot \omega^{-1}) \cdot \omega) 
= \lambda'(\beta) \;.
\end{align*}

That is, $\lambda' = \Lambda'|_{\mathcal{G}_x^x}$, such as we wanted to prove.
\end{proof}

Now we want to present the construction of  a set of homomorphisms of algebras $\mathfrak{T}(\mathcal{G})$ such that there exists an epimorphism from the groupoid $\mathcal{G}$ into $\mathfrak{T}(\mathcal{G})$ (when the last one is equipped with a suitable product). The above will be done under the assumption that the groupoid $\mathcal{G}$ is Hausdorff compact, which is presented as an extension of the so-called Tannaka's Lemma to the matter of topological groupoids.

Given $\Lambda \in \mathcal{A}(\mathcal{G})$ and a pair $u, v \in V_\Lambda$, we set $f^\Lambda_{u, v} \in \mathrm{C}(\mathcal{G})$ by the equation 
\[
f^\Lambda_{u, v}(\alpha) := (\Lambda(\alpha)(u), v) \;.
\] 

When $(\Lambda_{ij}(\alpha))_{n \times n}$ is a matrix representation of $\Lambda(\alpha)$, $v = (v_1, ... , v_n)$ and $u = (u_1, ... , u_n)$, it is not difficult to check that the following expression holds true
\begin{equation}
\label{f-matrix}
f^\Lambda_{u, v}(\alpha) = \sum_{i, j = 1}^n v_i \Lambda_{ij}(\alpha) u_j \;.
\end{equation}

An straightforward argument shows that $f \in \xi(\mathcal{G})$ if, and only if, $\overline{f} \in \xi(\mathcal{G})$. In fact, we have the following conditions:
\begin{enumerate}[a)]
\item Given $\Lambda \in \mathcal{A}(\mathcal{G})$, we have $f^{\Lambda^*}_{u, v} = \overline{f^\Lambda_{u, v}}$, where $\Lambda^*(x) = \Lambda(x) = V_\Lambda$ for each $x \in \mathcal{G}^{(0)}$ and $\Lambda^*(\alpha) = \Lambda(\alpha^{-1})$ for any $\alpha \in \mathcal{G}$ (in particular $\Lambda^* \in \mathcal{A}(\mathcal{G})$); 
\item $f^\Lambda_{e_i, e_j}(\alpha) = \Lambda_{ij}(\alpha)$ for any $\alpha \in \mathcal{G}$ and each $\Lambda \in \mathcal{A}(\mathcal{G})$;
\item We have $f^{\Lambda \oplus \Lambda'}_{u \oplus u', v \oplus v'} = f^\Lambda_{u, v} + f^{\Lambda'}_{u', v'}$ and $f^{\Lambda \otimes \Lambda'}_{u \otimes u', v \otimes v'} = f^\Lambda_{u, v} f^{\Lambda'}_{u', v'}$ for any pair $\Lambda, \Lambda' \in \mathcal{A}(\mathcal{G})$. 
\end{enumerate}

The former assumptions imply that the set of functions 
\begin{equation}
\label{xi(G)}
\xi(\mathcal{G}) := \{f^\Lambda_{u, v} :\; \Lambda \in \mathcal{A}(\mathcal{G}),\; u, v \in V_\Lambda \} \subset \mathrm{C}(\mathcal{G}) \;, 
\end{equation}
is a sub-algebra with unit of $\mathrm{C}(\mathcal{G})$, i.e., satisfying ${\bf 1} \in \xi(\mathcal{G})$. 

\begin{remark}
As a particular case of $\eqref{xi(G)}$, we obtain that $\xi(G)$ is well defined when $G$ is a topological group and also results in a sub-algebra of $\mathrm{C}(G)$ with the unit. In particular, we have that condition for any isotropy group $G := \mathcal{G}_a^a$. 
\end{remark}

Our next goal is to introduce a suitable definition of the {\it mean value function} $M_G$ associated with a Hausdorff compact group $G$ (similar to the one presented in \cite{MR1336382}). The existence of such a function is guaranteed in the following lemma.
\begin{lemma}
\label{mean-value-lemma}
Consider a compact group $G$. Then, there is a linear application $M_G : \mathrm{C}(G) \to \mathbb{C}$ satisfying the following conditions:
\begin{enumerate}[i)]
\item $M_G({\bf 1}) = 1$.
\item $M_G(f) \geq 0$ when $f \in \mathrm{C}^{\geq 0}$. Moreover, $M_G(f) = 0$ implies that $f \equiv 0$.
\item $|M_G(f)| \leq M_G(|f|)$.
\item Given $x, y \in G$, define $f_x(y) := f(xy)$, $f^x(y) := f(yx)$, we have
\begin{equation*}
M_G(f) = M_G(f^x) = M_G(f_x) = M_G(f^*) \;.
\end{equation*} 
\end{enumerate}
\end{lemma}
\begin{proof}
It is enough to choose the map $M_G$ satisfying the equation
\begin{equation*}
M_G(f) := \int_G f dm \;,
\end{equation*}
where $m \in \mathcal{M}_1(G)$ is a Haar measure.
\end{proof}

Observe that the map $M_G$ is continuous on $\mathrm{C}(G)$. The above, because it is linear and by item $ii)$ of Lemma \ref{mean-value-lemma}, we have $|M_G(f)| \leq M_G(\|f\|_\infty) = \|f\|_\infty$. 

Now we will present some properties of the mean value calculated on functions $f^\lambda_{u, v}$ belonging to $\xi(G)$, when $G$ is a compact group.

\begin{lemma}
\label{mv-lemma}
Consider a compact group $G$ and $\lambda \in \mathcal{A}(G)$ irreducible.
\begin{enumerate}
\item If $\lambda \sim \lambda_0$, then, $M_G(f^\lambda_{u, v}) = (u, v)$ for any pair $u, v \in V_\lambda$.
\item If $\lambda \nsim \lambda_0$, then, $M_G(f^\lambda_{u, v}) = 0$ for each $u, v \in V_\lambda$.
\end{enumerate}
\end{lemma}
\begin{proof}
Suppose that $\lambda \sim \lambda_0$, thus, we have $\lambda(x) = \mathrm{Id}_{V_\lambda}$ for each $x \in G$. By the above, it follows that $f^\lambda_{u, v}(x) = (u, v)$ for any pair $u, v \in V_\lambda$ which implies 
\begin{equation*}
M_G(f^\lambda_{u, v}) := \int_G f^\lambda_{u, v}(x) dm(x) = \int_G (u, v) dm = (u, v) \;.
\end{equation*}

On the other hand, assuming that $\lambda \nsim \lambda_0$. For each $v \in V_\lambda$, we define the linear map $L_v : V_\lambda \to \mathbb{C}$ by the expression $L_v(u) := M_G(f^\lambda_{u, v})$. Then, defining $f(x) = f^\lambda_{u, v}(x)$ for each $x \in G$, it follows that any $y \in G$ satisfies
\begin{align*}
L_v(\lambda(y)u) 
= M_G(f^\lambda_{\lambda(y)u, v}) 
&= M_G(f^y) \\
&= M_G(f^\lambda_{u, v}) = M_G(f) = L_v(u) \;.
\end{align*}

The former expression guarantees that the space $\mathrm{Ker}(L_v)$ is stable by $\lambda$. Furthermore, since $\lambda$ is irreducible, by Lemma \ref{reducible-lemma}, it follows that $\mathrm{Ker}(L_v) = \{0\}$ or $\mathrm{Ker}(L_v) = V_\lambda$ either. 

In the first case, we obtain that $\lambda(y)u = u$ for any $y \in G$ and each $u \in V_\lambda$ and, thus, $\lambda \sim \lambda_0$, which contradicts our assumption. 

So, we need to have $L_v \equiv 0$ such as we wanted to prove. 
\end{proof}

Our next goal is to present a characterization of surjective group homomorphisms in terms of the mean value function when suitable conditions are given.

\begin{lemma}
\label{sur-lemma}
Consider compact groups $G, H$, with $H$ Hausdorff, let $F : G \to H$ be a continuous groups homomorphism such that $M_G(g \circ F) = M_H(g)$ for any function $g \in \mathrm{C}(H)$. Then, $F$ is surjective.
\end{lemma}
\begin{proof}
Suppose that there exists $y \in H$ such that $y \notin F(G)$. Since $\{y\}$ and $F(G)$ is closed and $H$ is Hausdorff compact, by Urysohn's Lemma, there exists $g \in \mathrm{C}^{\geq 0}(H)$ such that $g|_{F(G)} \equiv 0$ and $g(y) = 1$. Therefore, it follows that $M_H(g) > 0 = M_G(g \circ F)$ which contradicts our hypothesis. 

By the above, we obtain that $F(G) = H$, and our claim holds.
\end{proof}

Consider $\lambda \in \mathcal{A}(G)$ and define
\begin{equation}
\label{eigenspace}
W_\lambda := \{u \in V_\lambda : \lambda(x)u = u \text{ for all } x \in G\} \;.
\end{equation}

In the following lemma, we present a characterization of the mean value map $M_G$ on functions belonging to $\xi(G)$.

\begin{lemma}
\label{dec-lemma}
Consider a compact group $G$, $\lambda \in \mathcal{A}(G)$ and $u, v \in V_\lambda$. Then, for any function $f^\lambda_{u, v} \in \xi(G)$, we have $M_G(f^\lambda_{u, v}) = (u', v)$, where $u'$ is the orthogonal projection of $u$ on $W_\lambda$. 
\end{lemma}
\begin{proof}
Let $\lambda = \lambda_1 \oplus ... \oplus \lambda_n$ be the decomposition of the unitary representation $\lambda$ in irreducible components, i.e., we have $V_\lambda = V_{\lambda_1} \oplus ... \oplus V_{\lambda_n}$. Without loss of generality, we can assume that $\lambda_i \sim \lambda_0$ for $i = 1, ... , k$ and $\lambda_i \nsim \lambda_0$ for $i = k+1, ... , n$ (i.e. the first $k$ components are identity representations). Then, by Lemma \ref{mv-lemma}, taking $u = u_1 + ... + u_n$ and $v = v_1 + ... + v_n$, with $u_j, v_j \in V_{\lambda_j}$, it follows that
\begin{equation*}
M_G(f^\lambda_{u, v}) = M_G\Bigl(\sum_{i=1}^n f^{\lambda_i}_{u_i, v_i}\Bigr) = \sum_{i=1}^n M_G(f^{\lambda_i}_{u_i, v_i}) \;.
\end{equation*}

On the other hand, given any $x \in G$, we have 
\begin{equation*}
\lambda(x)u = \sum_{i=1}^n \lambda_i(x)u_i = \sum_{i=1}^n u_i \;,
\end{equation*}
with $\lambda_i \nsim \lambda_0$ for $i = k+1, ... , n$. Then, $\lambda_i(x)u_i = u_i$ for all $x \in G$ implies that $u_i = 0$ when $i = k+1, ... , n$. 

By the above, it follows that $W_\lambda = V_{\lambda_1} \oplus ... \oplus V_{\lambda_k}$ and $u' = u_1 + ... + u_k$. Besides that, by Lemma \ref{mv-lemma}, we have
\begin{equation*}
M_G(f^\lambda_{u, v}) = \sum_{i=1}^k (u_i, v_i) = (u', v) \;,
\end{equation*}
such as we wanted to prove.
\end{proof}

Given a family $\mathcal{F} \subset \mathcal{A}(G)$, we say that $\mathcal{F}$ is {\it complete} when satisfies the following conditions
\begin{enumerate}[i)]
\item $\lambda_0 \in \mathcal{F}$;
\item For any pair $\lambda_1, \lambda_2 \in \mathcal{F}$, we have $\lambda_1 \oplus \lambda_2 \in \mathcal{F}$ and $\lambda_1 \otimes \lambda_2 \in \mathcal{F}$.
\item $\overline{\lambda} \in \mathcal{F}$ when $\lambda \in \mathcal{F}$.
\end{enumerate}

Next, we characterize the surjectivity of continuous group homomorphisms in terms of complete families of unitary representations on compact groups. 

\begin{lemma}
\label{sur-main-lemma}
Consider $G$ and $H$ compact groups, with $H$ Hausdorff, and let $F : G \to H$ be a continuous group homomorphism. Assume that there exists a complete family $\mathcal{F} \subset \mathcal{A}(G)$ such that for each $\lambda \in \mathcal{F}$ there is $\lambda' \in \mathcal{A}(H)$ such that $V_\lambda = V_{\lambda'} = V$ and the following conditions hold true:
\begin{enumerate}[i)]
\item $\lambda_0 = \lambda'_0$;
\item If $\lambda_1, \lambda_2 \in \mathcal{F}$, then, $(\lambda_1 \oplus \lambda_2)' = \lambda'_1 \oplus \lambda'_2$ and $(\lambda_1 \otimes \lambda_2)' = \lambda'_1 \otimes \lambda'_2$;
\item $\overline{\lambda'} = (\overline{\lambda})'$ for any $\lambda \in \mathcal{F}$;
\item $\lambda(x) = \lambda'(F(x))$ for each $x \in G$;
\item If $y_1, y_2 \in H$ and $y_1 \neq y_2$, then, there are $\lambda \in \mathcal{F}$ and $u, v \in V$ such that $f^{\lambda'}_{u, v}(y_1) \neq f^{\lambda'}_{u, v}(y_2)$;
\item For any $\lambda \in \mathcal{F}$, there are $u, v \in V$ such that $M_H(f^{\lambda'}_{u, v}) = M_G(f^\lambda_{u, v})$.
\end{enumerate}

Then, the map $F$ is surjective. 
\end{lemma}
\begin{proof}
Since $\mathcal{F}$ is a complete family of unitary representations, by conditions $i)$, $ii)$ and $iii)$, it follows that $\mathcal{F}' := \{\lambda' : \lambda \in \mathcal{F}\} \subset \mathcal{A}(H)$ is also a complete family of unitary representations. By the above, we obtain that 
\begin{equation*}
\xi'(G) := \{f^{\lambda'}_{u, v} : \lambda \in \mathcal{F}, u, v \in V\}\;,
\end{equation*} 
is a sub-algebra with unit of $\mathrm{C}(H)$ such that $f \in \xi'(G)$ implies $\overline{f} \in \xi'(G)$. Furthermore, by condition $v)$ and the Stone-Weierstrass' theorem, it follows that the set $\xi'(G)$ is dense in $\mathrm{C}(H)$. Besides that, by conditions $iv)$ and $vi)$ we have $M_G(g \circ F) = M_H(g)$ for any $g \in \xi'(G)$. So, by continuity of the mean value function, we obtain that $M_G(g \circ F) = M_H(g)$ for any $g \in \mathrm{C}(H)$ which implies that $F$ is surjective by Lemma \ref{sur-lemma}. 
\end{proof}

The former lemma also implies that for any $\lambda \in \mathcal{F}$, each $u, v \in V_\lambda$ and any $x \in G$, the following expression holds true
\begin{equation*}
f^{\lambda'}_{u, v}(F(x)) = (\lambda'(F(x))u, v) = (\lambda(x)u, v) = f^\lambda_{u, v}(x) \;.
\end{equation*}

Consider the set of homomorphisms of algebras
\begin{equation*}
\mathfrak{T}(\mathcal{G}) := \{ T : \xi(\mathcal{G}) \to \mathbb{C} :\; T({\bf 1}) = 1 \text{ and } T(\overline{f})=\overline{T(f)} \} \;.
\end{equation*}

Note that given any $\alpha \in \mathcal{G}$, the map $T_\alpha : \xi(\mathcal{G}) \to \mathbb{C}$ given by the equation $T_\alpha(f) := f(\alpha)$ belongs to the set $\mathfrak{T}(\mathcal{G})$, because $T_\alpha({\bf 1}) = {\bf 1}(\alpha) = 1$ and for any $f \in \xi(\mathcal{G})$ we have $T_\alpha(\overline{f}) = \overline{f}(\alpha) = \overline{f(\alpha)} = \overline{T_\alpha(f)}$. So, the set $\mathfrak{T}(\mathcal{G})$ is non-empty.

Given $T \in \mathfrak{T}(\mathcal{G})$ and any $\Lambda \in \mathcal{A}(\mathcal{G})$, we define the linear map $L^\Lambda_T : V_\Lambda \to V_\Lambda$ by the equation
\begin{equation}
\label{L-T-map}
(L^\Lambda_T(u), v) := T(f^\Lambda_{u, v}) \;,
\end{equation}
where $u, v \in V_\Lambda$. It is not difficult to check that the matrix representation of $L^\Lambda_T$ with respect to the orthonormal basis $\{e_1, ... , e_n\}$ is of the form $(T(f^\Lambda_{e_i, e_j}))_{n \times n}$. Moreover, by \eqref{L-T-map} and conditions $a)$, $b)$ and $c)$ stated above, it follows that: 
\begin{enumerate}[a')]
\item $L^{\Lambda_0}_T = \mathrm{Id}_{V_{\Lambda_0}}$. 
\item We have  $L^{\Lambda \oplus \Lambda'}_T = L^\Lambda_T \oplus L^{\Lambda'}_T$, $L^{\Lambda \otimes \Lambda'}_T = L^\Lambda_T \otimes L^{\Lambda'}_T$ and $L^{\Lambda^*}_T = \overline{L^\Lambda_T}$, for any pair $\Lambda, \Lambda' \in \mathcal{A}(\mathcal{G})$.
\end{enumerate}

In particular, the above implies $L^\Lambda_T \in \mathrm{U}(V_\Lambda)$.
\begin{lemma}
\label{wd-lemma}
Consider $\Lambda, \Lambda' \in \mathcal{A}(\mathcal{G})$, $u, v \in V_\Lambda$ and $u', v' \in V_{\Lambda'}$. Assume that $f^\Lambda_{u, v}(\alpha) = f^{\Lambda'}_{u', v'}(\alpha)$ for each $\alpha \in \mathcal{G}$. Then, $(L^\Lambda_T(u), v) = (L^{\Lambda'}_T(u'), v')$ for any $T \in \mathfrak{T}(\mathcal{G})$.
\end{lemma}
\begin{proof}
Take $\widetilde{\Lambda} \in \mathcal{A}(\mathcal{G})$ and assume that $f^{\widetilde{\Lambda}}_{\widetilde{u}, \widetilde{v}} \equiv 0$. Then, given any $T \in \mathfrak{T}(\mathcal{G})$, we obtain that $(L^{\widetilde{\Lambda}}_T(\widetilde{u}), \widetilde{v}) = T({\bf 0})$. Since $T({\bf 0}) = T({\bf 0})^2$ and $T({\bf 0}) = 0 \neq 1$, it follows that $(L^{\widetilde{\Lambda}}_T(\widetilde{u}), \widetilde{v}) = 0$. 

By hypothesis, we have $f^{\Lambda \oplus \Lambda'}_{u \oplus u', v \oplus (-v')} \equiv 0$. So, taking $\widetilde{\Lambda} = \Lambda \oplus \Lambda'$, $\widetilde{u} = u \oplus u'$ and $\widetilde{v} = v \oplus (-v')$, we obtain that $(L^\Lambda_T(u), v) = (L^{\Lambda'}_T(u'), v')$ such as we wanted to prove. 
\end{proof}

In the next lemma, we prove that any unitary transformation satisfying suitable conditions can be characterized in terms of elements belonging to $\mathfrak{T}(\mathcal{G})$ as a linear transformation satisfying the expression in \eqref{L-T-map}. 

\begin{lemma}
\label{product-thm}
Assume that for each $\Lambda \in \mathcal{A}(\mathcal{G})$ there exists $L^\Lambda \in \mathrm{U}(V_\Lambda)$ satisfying the properties $a')$ and $b')$ stated above. Then, there is $T \in \mathfrak{T}(\mathcal{G})$ such that $L^\Lambda_T = L^\Lambda$.
\end{lemma}
\begin{proof}
Given $L^\Lambda \in \mathrm{U}(V_\Lambda)$, define $T : \xi(\mathcal{G}) \to \mathbb{C}$ by $T(f^\Lambda_{u, v}) := (L^\Lambda(u), v)$ for each $u, v \in V_\Lambda$. By Lemma \ref{wd-lemma} the map $T$ is well defined and satisfies \eqref{L-T-map}. 

Besides that, conditions $a')$ and $b')$ stated above imply that $T$ is a linear homomorphism such that $T({\bf 1}) = 1$ and $T(\overline{f}) = \overline{T(f)}$, i.e., $T \in \mathfrak{T}(\mathcal{G})$.
\end{proof}

We induce the weak topology on $\mathfrak{T}(\mathcal{G})$ which is generated by neighborhoods of the unit $T_a$ as basic open sets of the form
\begin{equation*}
\{T \in \mathfrak{T}(\mathcal{G}) : |T(f_i) - T_a(f_i)| < \epsilon_i, a \in \mathcal{G}^{(0)}\} \;.
\end{equation*}

Furthermore, we have that $\mathfrak{T}(\mathcal{G}) \subset \prod_{f \in \xi(\mathcal{G})}\{z \in \mathbb{C} : |z| \leq \sup_{T \in \mathfrak{T}(\mathcal{G})}|T(f)|\}$  and the set $\prod_{f \in \xi(\mathcal{G})}\{z \in \mathbb{C} : |z| \leq \sup_{T \in \mathfrak{T}(\mathcal{G})}|T(f)|\}$ is compact by Tychonoff's theorem, because for any $f \in \xi(\mathcal{G})$ we have $\sup_{T \in \mathfrak{T}(\mathcal{G})}|T(f)| \leq \|f\|_\infty$. Moreover, by continuity of the conjugation map, it follows that $\mathfrak{T}(\mathcal{G})$ is closed and, thus, compact with the weak topology.

Our next goal is to define a product on $\mathfrak{T}(\mathcal{G})$ which induces a structure of groupoid on that set of linear homomorphisms under suitable conditions.

Given $T_1, T_2 \in \mathfrak{T}(\mathcal{G})$, we define the product $T_1 \cdot T_2$ by the expression
\begin{equation*}
L^\Lambda_{T_1 \cdot T_2} := L^\Lambda_{T_1} L^\Lambda_{T_2} \;,
\end{equation*}
where $\Lambda \in \mathcal{A}(\mathcal{G})$ is arbitrary. Observe that the product in \eqref{product-T} is well defined by Lemma \eqref{product-thm}, which also implies that $T_1 \cdot T_2 \in \mathfrak{T}(\mathcal{G})$. 

Next, we introduce the definition of a set of homomorphisms of algebras $\mathfrak{G}$ such that there exists an epimorphism from the groupoid $\mathcal{G}$ into $\mathfrak{G}$. Define
\begin{equation*}
I_a^b := \{f \in \xi(\mathcal{G}) : f|_{\mathcal{G}_a^b} \equiv 0\} \;.
\end{equation*}

It is not difficult to check that $I_a^b$ is an ideal of $\xi(\mathcal{G})$ for any pair $a, b \in \mathcal{G}^{(0)}$. Using that, we consider 
\begin{equation*}
\mathfrak{G} := \{(a, b, T) :\; a, b \in \mathcal{G}^{(0)},\; T \in \mathfrak{T}(\mathcal{G}) \text{ and } T(f) = 0 \; \forall f \in I_a^b\} \;,
\end{equation*}
equipped with the product topology and
\begin{equation*}
\mathfrak{G}_a^b := \{(a, b, T) :\; T \in \mathfrak{T}(\mathcal{G}) \text{ and } T(f) = 0 \; \forall f \in I_a^b\} \;.
\end{equation*}

Then, we have that each $\mathfrak{G}_a^b$ is a closed subspace of $\mathfrak{G}$ and 
\begin{equation*}
\mathfrak{G} := \bigcup_{a, b \in \mathcal{G}^{(0)}} \mathfrak{G}_a^b \;.
\end{equation*}

Observe that
\begin{align*}
\mathfrak{G} = \{(a, b, T) 
:&\; a, b \in \mathcal{G}^{(0)},\; T \in \mathfrak{T}(\mathcal{G}) \text{ and } (L^\Lambda_T(u), v) = 0,\; \forall \Lambda \in \mathcal{A},\\ 
&\forall u, v \in V_\Lambda \text{ such that } (\Lambda(\alpha)u, v) = 0,\; \text{ if } \alpha \in \mathcal{G}_x^y \}\; .
\end{align*}

The above, because $f^\Lambda_{u, v}(\alpha) = (\Lambda(\alpha)u, v)$ for all $\alpha \in \mathcal{G}$. Given $(a, b, T_1) \in \mathfrak{G}_a^b$ and $(b, c, T_2) \in \mathfrak{G}_b^c$, define
\begin{equation}
\label{product-T}
(a, b, T_1) \cdot (b, c, T_2) := (a, c, T_1 \cdot T_2) \;.
\end{equation}

Observe that the product $\cdot$ defined in \eqref{product-T} results continuous by the expression in \eqref{L-T-map}. In the next lemma, we prove that such a product induces a structure of topological groupoid on $\mathfrak{T}(\mathcal{G})$.

\begin{lemma}
Assume that $\mathcal{G}$ is a compact groupoid. Then, the product in \eqref{product-T} induces a structure of topological groupoid on the set $\mathfrak{T}(\mathcal{G})$.
\end{lemma}
\begin{proof}
Consider $a \in \mathcal{G}^{(0)}$, then, $T_a(f^\Lambda_{u, v}) = (\Lambda(a)u, v) = (u, v)$ which implies that $L^\Lambda_{T_a} = \mathrm{Id}_{V_\Lambda}$ for any $\Lambda \in \mathcal{A}(\mathcal{G})$. So, by \eqref{product-T}, it follows that 
\begin{equation*}
(b, a, T_1) \cdot (a, a, T_a) = (b, a, T_1) \text{ and }
(a, a, T_a) \cdot (a, b, T_2) = (a, b, T_2) \;,
\end{equation*}
for any $T_1, T_2 \in \mathfrak{T}(\mathcal{G})$ and each $a, b \in \mathcal{G}^{(0)}$, i.e., $(a, a, T_a)$ is a unit of $\mathfrak{G}$ for each $a \in \mathcal{G}^{(0)}$. 

On the other hand, since $L^\Lambda_T \in \mathrm{U}(V_\Lambda)$ for each $T \in \mathfrak{T}(\mathcal{G})$, by Theorem \ref{product-thm}, we are able to define $T^{-1}$ as the only element in $\mathfrak{T}(\mathcal{G})$ associated to $(L^\Lambda_T)^{-1}$. Then, defining $(a, b, T)^{-1} := (b, a, T^{-1})$, we obtain that
\begin{equation*}
(b, a, T_1) \cdot (b, a, T_1)^{-1} = (b, b, \mathrm{Id}_{V_\Lambda}) \text{ and }
(a, b, T_2) \cdot (a, b, T_2)^{-1} = (a, a, \mathrm{Id}_{V_\Lambda})
\end{equation*}
for any $T_1, T_2 \in \mathfrak{T}(\mathcal{G})$ and each $a, b \in \mathcal{G}^{(0)}$. So, by the expression in \eqref{L-T-map} the map $^{-1}$ is continuous.

The associativity of the product in \eqref{product-T} follows immediately of the associativity that satisfies the composition of linear transformations.

Besides that, given any pair $(\alpha, \beta) \in \mathcal{G}^{(2)}$, we have
\begin{align*}
((L_{T_\alpha}^\Lambda L_{T_\beta}^\Lambda)(u), v)
&= (L_{T_\beta}^\Lambda(u), (L_{T_\alpha}^\Lambda)^*(v)) \\ 
&= (\Lambda(\beta)(u), (L_{T_\alpha}^\Lambda)^*(v)) \\
&= (L_{T_\alpha}^\Lambda(\Lambda(\beta)(u)), v) \\
&= (\Lambda(\alpha)(\Lambda(\beta)(u)), v) \\
&= (\Lambda(\alpha \cdot \beta)(u), v)
= (L_{T_{\alpha \cdot \beta}}(u), v) \;.
\end{align*}

That is, $T_\alpha \cdot T_\beta = T_{\alpha \cdot \beta}$. In other words, the elements $(r(\alpha), s(\alpha), T_\alpha)$ and $(r(\beta), s(\beta), T_\beta)$ can be composed if, and only if, $(\alpha, \beta) \in \mathcal{G}^{(2)}$. In fact, the former property guarantees that $\mathfrak{G}$ has at least as many units as the groupoid $\mathcal{G}$. 

The above guarantees that conditions $i)$, $ii)$, $iii)$, and $iv)$ stated at the beginning of Section \ref{topological-groupoids-section} hold, such as we wanted to prove.
\end{proof}

In the last one of the results presented in this section we prove a version of the so-called Tannaka's Theorem. The statement of the result is the following one.

\begin{theorem}
Suppose that $\mathcal{G}$ is a Hausdorff connected compact groupoid satisfying the hypothesis {\bf (H)} and consider the map $\phi : \mathcal{G} \to \mathfrak{G}$, given by the equation
\begin{equation*}
\phi(\alpha) := (r(\alpha), s(\alpha), T_\alpha) \;.
\end{equation*} 

Then, $\phi$ is an epimorphism of topological groupoids.
\end{theorem}
\begin{proof}
Fix $a \in \mathcal{G}^{(0)}$ and define $\mathfrak{G}_a^a := \{(a, a, T_\alpha) : \alpha \in \mathcal{G}_a^a\}$, by Lemma \ref{lemma-H2}, it is enough to show that the function $\phi |_{\mathcal{G}_a^a} : \mathcal{G}_a^a \to \mathfrak{G}_a^a$ is surjective.

Under the notation that appears in Lemma \ref{sur-main-lemma}, we consider $G := \mathcal{G}_a^a$, $F := \phi$ and $H := \{T_\alpha : \alpha \in \mathcal{G}_a^a\}$. Besides that, we define $\mathcal{F} := \{\Lambda|_{\mathcal{G}_a^a} : \Lambda \in \mathcal{A}(\mathcal{G}) \}$. 

It follows immediately that $\mathcal{F}$ is a complete family of unitary representations of $G$. Take $\Lambda_1, \Lambda_2 \in \mathcal{A}(\mathcal{G})$ such that $\Lambda_1|_G = \Lambda_2|_G$ and $V_{\Lambda_1} = V_{\Lambda_2} = V$. Since $f^{\Lambda_1}_{u, v}|_G = f^{\Lambda_2}_{u, v}|_G$, it follows that each $T \in H$ satisfies
\begin{equation*}
T(f^{\Lambda_1}_{u, v}) - T(f^{\Lambda_2}_{u, v}) = T(f^{\Lambda_1}_{u, v} - f^{\Lambda_2}_{u, v}) = 0 \;,
\end{equation*} 

So, for any $\lambda \in \mathcal{F}$, we can define $\lambda' \in \mathcal{A}(H)$ by the expression
\begin{equation*}
\lambda'(T) := L^\Lambda_T \;,
\end{equation*}
where $\Lambda \in \mathcal{A}(\mathcal{G})$ satisfies $\Lambda|_G = \lambda$ and $V_\Lambda = V_\lambda = V_{\lambda'}$.

Conditions $i)$, $ii)$ and $iii)$ of Lemma \ref{sur-main-lemma} follow from the properties $a')$ and $b')$ of the map $L^\Lambda_T$ stated above. Besides that, given $\Lambda \in \mathcal{A}(\mathcal{G})$, $\lambda = \Lambda|_G$, $a \in G$ and $u, v \in V_\Lambda$, it follows that
\begin{align*}
(\lambda'(F(a))u, v) 
= (L^\Lambda_{F(a)}(u), v) 
&= F(a)(f^\Lambda_{u, v}) \\
&= T_a(f^\Lambda_{u, v}) \\
&= f^\Lambda_{u, v}(a) 
= (\Lambda(a)u, v) 
= (\lambda(a)u, v) \;. 
\end{align*}

The above proves that $iv)$ of Lemma \ref{sur-main-lemma} holds. Besides that, for any $\Lambda \in \mathcal{A}(\mathcal{G})$, each $T \in H$ and any $u, v \in V_\Lambda$, we have
\begin{equation*}
T(f^\Lambda_{u, v}) = (L^\Lambda_T(u), v) = (\lambda'(T)u, v) = f^{\lambda'}_{u, v}(T) \;,
\end{equation*}
where $\lambda = \Lambda|_G$. So, condition $v)$ of Lemma \ref{sur-main-lemma} holds true.

In order to finish this proof, it is enough to show that $W_\lambda = W_{\lambda'}$ which implies that the condition $vi)$ of Lemma \ref{sur-main-lemma} holds. The above, as a consequence of Lemma \ref{dec-lemma}. Indeed, since $\lambda = \Lambda|_G$, it follows that
\begin{align*}
W_\lambda 
&= \{u \in V_\lambda : \lambda(a)u = u \text{ for all } a \in G \} \\
&= \{u \in V_\Lambda : \Lambda(a)u = u \text{ for all } a \in G \} \\
&= \{u \in V_\Lambda : f^\Lambda_{u, v}(a) = (u, v) \text{ for all } a \in G, \; v \in V_\Lambda \} \;.
\end{align*}

On the other hand, we have that 
\begin{align*}
W_{\lambda'}
&= \{u \in V_{\lambda'} : \lambda'(T)(u) = u \text{ for all } T \in H \} \\
&= \{u \in V_\Lambda : (L^\Lambda_T(u), v) = (u, v) \text{ for all } T \in H,\; v \in V_\Lambda \} \\
&= \{u \in V_\Lambda : T(f^\Lambda_{u, v}) = (u, v) \text{ for all } T \in H,\; v \in V_\Lambda \} \;.
\end{align*}

If $u \in W_\lambda$, it follows that $f^\lambda_{u, v}(a) = (u, v)$ for all $a \in G$ and any $v \in V_\lambda$ which implies that $f^\Lambda_{u, v}(a) = (u, v)$ because $f^\Lambda_{u, v}|_G = f^\lambda_{u, v}$. By the above, we obtain that $T(f^\Lambda_{u, v}) = (u, v)$ for any $T \in H$ and each $v \in V_\Lambda$. Then, it follows that $u \in W_{\lambda'}$.

On the other hand, if $u \in W_{\lambda'}$, we obtain that $T(f^\Lambda_{u, v}) = (u, v)$ for any $T \in H$ and all $v \in V_\Lambda$. In particular, the above implies that $f^\Lambda_{u, v}(a) = T_a(f^\Lambda_{u, v}) = (u, v)$ for each $a \in G$ and any $v \in V_\Lambda$. So, we obtain that $u \in W_\lambda$.
\end{proof}

\medskip

No new data were created or analysed in this study.


\begin{thebibliography}{10}


\bibitem{Ibort2}
F.~M. Ciaglia, F.~Di~Cosmo, A.~Ibort, and G.~Marmo.
\newblock  Schwinger's picture of quantum mechanics. 
\newblock {\em arXiv (2020)}

\bibitem{Ibort3}  
F.~M. Ciaglia, A.~Ibort, and G.~Marmo.
\newblock  Schwinger's picture of quantum mechanics: Groupoids.
\newblock {\em Int. J. Geom. Methods Mod. Phys.}, Vol. 16, No. 08, 1950119 (2019)

\bibitem{Ibort1}
F.~M. Ciaglia, A.~Ibort, and G.~Marmo.
\newblock Schwinger's picture of quantum mechanics {III}: The statistical interpretation.
\newblock {\em Int. J. Geom. Methods Mod. Phys.}, Vol. 16, No. 11, 1950165 (37 pages) (2019) 

\bibitem{Ibort}
F. M. Ciaglia, F. Di Cosmo, A. Ibort and G. Marmo. Quantum tomography and Schwinger's picture of quantum mechanics. J. Phys. A, 55, no. 27, Paper No. 274008, 23 pp  (2022)


\bibitem{MR548112}
A.~Connes.
\newblock Sur la th\'{e}orie non commutative de l'int\'{e}gration.
\newblock In {\em Alg\`ebres d'op\'{e}rateurs ({S}\'{e}m., {L}es
  {P}lans-sur-{B}ex, 1978)}, volume 725 of {\em Lecture Notes in Math.}, pages
  19--143. Springer, Berlin, 1979.

\bibitem{MR4378494}
G.~G. de~Castro, A.~O. Lopes, and G.~Mantovani.
\newblock Haar systems, {KMS} states on von {N}eumann algebras and
  {$C^\ast$}-algebras on dynamically defined groupoids and noncommutative
  integration.
\newblock In {\em Modeling, dynamics, optimization and bioeconomics {IV}},
  volume 365 of {\em Springer Proc. Math. Stat.}, pages 79--137. Springer,
  Cham, [2021] \copyright 2021.

\bibitem{zbMATH07092992}
A.~{Ibort} and M.~A. {Rodr\'{\i}guez}.
\newblock {\em {An introduction to groups, groupoids and their
  representations}}.
\newblock Boca Raton, FL: CRC Press/Science Publishers, 2020.

\bibitem{LCH} B-S Lin, H-L Chen and T- Heng,
Note on Connes spectral distances of qubits, arXiv (2020)


\bibitem{zbMATH03840637}
D.~{Kastler}.
\newblock {On A. Connes' noncommutative integration theory}.
\newblock {\em {Commun. Math. Phys.}}, 85:99--120, 1982.

\bibitem{MR2215776}
A.~Kumjian and J.~Renault.
\newblock K{MS} states on {$C^*$}-algebras associated to expansive maps.
\newblock {\em Proc. Amer. Math. Soc.}, 134(7):2067--2078, 2006.


\bibitem{Lat} 
F. Latremoliere,
\newblock {On Quantum groupoids},
\newblock {Notes on line - University of California, Berkeley}


\bibitem{zbMATH07352601}
A.~O. {Lopes} and J.~K. {Mengue}.
\newblock {Thermodynamic formalism for Haar systems in noncommutative
  integration: transverse functions and entropy of transverse measures}.
\newblock {\em {Ergodic Theory Dyn. Syst.}}, 41(6):1835--1863, 2021.

\bibitem{MR3993188}
A.~O. Lopes and E.~R. Oliveira.
\newblock Continuous groupoids on the symbolic space, quasi-invariant
  probabilities for {H}aar systems and the {H}aar-{R}uelle operator.
\newblock {\em Bull. Braz. Math. Soc. (N.S.)}, 50(3):663--683, 2019.




\bibitem{Ni} 
D. Nikshych,
\newblock {A duality theorem for quantum groupoids}. 
\newblock {New trends in Hopf algebra theory (La Falda, 1999), 237-243, Contemp. Math., 267, Amer. Math. Soc., Providence, RI, 2000}

\bibitem{Put} I. Putnam, Lecture Notes on $C^*$-algebras (2019)

\bibitem{Ren1} J. Renault, A Groupoid approach to $C^*$-algebras, Lecture Notes in
Mathematics 793, Springer-Verlag, (1980)

\bibitem{MR2536186}
J.~Renault.
\newblock {\em {$C^\star$}-algebras and dynamical systems}.
\newblock 27${^{{}}{\rm{o}}}$ Col\'{o}quio Brasileiro de Matem\'{a}tica.



\bibitem{Scha} 
P. Schaueberg, Tannaka duality for arbitrary Hopf algebras, Algebra Report, 66. Verlag Reinhard Fischer (1992).

\bibitem{Tan}   T. Tannaka, Dualitat der nicht-kommutativen bikompakten gruppen, Thohoku Math. Jour,  58 pp1-12 (1938)


\bibitem{MR14095}
K.~Yosida.
\newblock On the duality theorem of non-commutative compact groups.
\newblock {\em Proc. Imp. Acad. Tokyo}, 19:181--183, 1943.

\bibitem{MR1336382}
K.~Yosida.
\newblock {\em Functional analysis}.
\newblock Classics in Mathematics. Springer-Verlag, Berlin, 1995.
\newblock Reprint of the sixth (1980) edition.


\end{thebibliography}
\end{document}